\documentclass{article}

\usepackage{cite}
\usepackage{amsmath,amssymb,amsfonts}
\usepackage{graphicx}
\usepackage{textcomp}
\usepackage{xcolor}
\def\BibTeX{{\rm B\kern-.05em{\sc i\kern-.025em b}\kern-.08em
    T\kern-.1667em\lower.7ex\hbox{E}\kern-.125emX}}

\usepackage[english]{babel}
\usepackage[utf8]{inputenc}
\usepackage[T1]{fontenc}
\usepackage{lmodern}
\usepackage{braket}
\usepackage{wrapfig}
\usepackage{graphicx}
\usepackage{amsmath,amssymb,amsfonts,verbatim}
\usepackage{hyperref}
\usepackage{mathtools}
\usepackage{enumitem}
\usepackage{geometry}

\usepackage{natbib}
\usepackage{xr}

\usepackage{amsmath,amssymb,amsfonts,verbatim}
\usepackage{graphicx}
\usepackage[utf8]{inputenc}
\usepackage{textcomp}
\usepackage{xcolor}
\usepackage{braket}
\usepackage{soul}
\usepackage[colorinlistoftodos]{todonotes}

\usepackage{caption}
\usepackage{subcaption}

\usepackage{enumitem}

\usepackage{cite}
\usepackage{amsmath,amssymb,amsfonts,verbatim}
\usepackage{amsthm}
\usepackage{graphicx}
\usepackage[utf8]{inputenc}
\usepackage{textcomp}
\usepackage{xcolor}
\usepackage{braket}
\usepackage{soul}
\usepackage{enumitem}
\usepackage[colorinlistoftodos]{todonotes}

\usepackage{algpseudocode}
\usepackage{algorithm}

\usepackage{wrapfig}

\usepackage{caption}

\newtheorem{definition}{Definition}
\newtheorem{theorem}{Theorem}
\newtheorem{corollary}{Corollary}

\usepackage{mdframed}

\usepackage{tikz}
\usetikzlibrary{positioning}
\usetikzlibrary{fit}
\usetikzlibrary{calc}
 
\tikzstyle{pinstyle} = [pin edge={to-,thin,black}]
 
\allowdisplaybreaks

\newcommand{\nat}{\mathbb{N}}
\newcommand{\reals}{\mathbb{R}}

\newcommand{\dataset}{\boldsymbol{\beta}}

\newcommand{\gaussN}{\mathcal{N}}

\newcommand{\Prob}{\mathbb{P}}

\newcommand{\barbeta}{\bar{\boldsymbol{\beta}}}
\newcommand{\optstat}{\Phi^*(\bar{\boldsymbol{\beta}})}
\newcommand{\optstati}{\hat{\Phi}^i(\barbeta)}
\newcommand{\stati}{\Phi^i}

\newcommand{\fcdf}{F_{\rvtest}}
\newcommand{\fccdf}{\bar{F}_{\rvtest}}
\newcommand{\Popt}{\mathbb{P}_{\optstat}}

\newcommand{\rvtest}{\Psi}
\newcommand{\rvtesti}{\Psi^i}

\newcommand{\numagents}{M}

\newcommand{\conset}{X_c}
\newcommand{\effset}{X_{PF}}
\newcommand{\simplex}{\mathcal{W}_{\numagents}}
\newcommand{\psimplex}{\simplex^+}
\newcommand{\woptset}{S(\sw)}

\newcommand{\lconst}{\alpha_{\dtime}}

\newcommand{\dtime}{t}
\newcommand{\respi}{\var^i_{\dtime}}
\newcommand{\nrespi}{\tilde{\var}^i_{\dtime}}
\newcommand{\noise}{\epsilon}
\newcommand{\noisei}{\noise^i}
\newcommand{\ndistr}{\Lambda}

\newcommand{\sw}{\mu}

\newcommand{\var}{\beta}
\newcommand{\varti}{\var_t^i}

\newcommand{\dimn}{N}

\newcommand{\numadv}{m}

\newcommand{\statei}{x}
\newcommand{\measi}{y^i}
\newcommand{\snoisei}{w}
\newcommand{\mnoisei}{v^i}
\newcommand{\stime}{t} 
\newcommand{\ftime}{k} 
\newcommand{\sdim}{q} 
\newcommand{\mdim}{p} 
\newcommand{\sncov}{Q_{\stime}(\varti)}

\newcommand{\mncovi}{R_{\stime}(\varti)}
\newcommand{\argo}{\beta}

\newcommand{\kstate}{\hat{x}} 
\newcommand{\kcov}{\Sigma} 
\newcommand{\ARE}{\mathcal{A}(\lconst,\varti,\kcov)}

\title{Multi-Agent Inverse Learning for Sensor Networks: Identifying Coordination in UAV Networks \thanks{This research was supported by NSF grants CCF-2312198 and CCF-2112457 and U. S. Army Research Office under grant W911NF-24-1-0083}}

\author{Luke Snow,  Vikram Krishnamurthy \thanks{Department of Electrical \& Computer Engineering, Cornell University, Ithaca, NY 14853, USA.  emails: las474@cornell.edu and vikramk@cornell.edu}}

\begin{document}



\maketitle

\begin{abstract}
Suppose there is an adversarial UAV network being tracked by a radar. How can the radar determine whether the UAVs are \textit{coordinating} in some well-defined sense? How can the radar infer the objectives of the individual UAVs and the network as a whole? 
We present an abstract interpretation of such a strategic interaction, allowing us to conceptualize coordination as a linearly constrained multi-objective optimization problem. Then, we present some tools from microeconomic theory that allow us to detect coordination and reconstruct individual UAV objective functions, from radar tracking signals. This corresponds to performing \textit{inverse multi-objective optimization}. We present details for how the abstract microeconomic interpretation corresponds to, and naturally arises from, physical-layer radar waveform modulation and multi-target filtering. This article serves as a tutorial, bringing together concepts from several established research contributions in an expository style. 
\end{abstract}

\section{Introduction}
In strategic environments, autonomous systems such as UAVs are becoming ubiquitous for reconnaissance, surveillance, and combative purposes. Often such autonomous systems are deployed in groups, e.g., UAV swarms, in order to collect information more efficiently or to multiply the combative force. Furthermore, these multi-agent intelligent systems typically have sophisticated sensors and communication capabilities which allow them to respond in real-time to an adversary's probe, e.g., radar tracking signals. This results in a strategic interaction between the multi-agent system and the adversary; the study of this interaction at the physical layer, for instance analyzing electromagnetic suppression techniques, is typically referred to as 'electronic warfare'. 

We consider a multi-agent strategic interaction scenario, in which a radar is tracking a network of UAVs. We take the perspective of the radar, and ask how can we detect \textit{coordination} in the UAV network? Such coordination detection would not only allow us to understand the functionality of the network, but when combined with estimates for the UAV objectives would allow us to \textit{predict future network behavior}. Thus, the second question we ask is: if the network is coordinating, how can we reconstruct individual objective functions which induce the observed aggregate behavior? 

We study this problem at a higher level of abstraction than traditional electronic warfare investigations; this allows us to formulate the 'coordination' problem as a general linearly-constrained multi-objective optimization. Then, the problem of detecting coordination and reconstructing feasible objective functions becomes that of \textit{inverse multi-objective optimization}. We present several tools from microeconomic theory which allow us to accomplish this inverse learning problem efficiently. While this microeconomic interpretation is conceptualized at a higher level of abstraction than traditional electronic warfare procedures, we also present how this framework arises naturally from physical-layer considerations such as radar waveform modulation and multi-target filtering algorithms. 

This chapter is organized as follows. Section~\ref{sec:MOO} presents the mathematical details of (forward and inverse) multi-objective optimization, and presents the microeconomic tools which can be used to accomplish general inverse multi-objective optimization. Then, Section~\ref{sec:UAV} presents the UAV network coordination detection procedure. First, the radar - UAV network interaction dynamics are specified, then it is shown how the microeconomic interpretation arises from filtering-level tracking considerations. Finally, in Section~\ref{sec:cdet} we present the application of the microeconomic tools from Section~\ref{sec:MOO} to the coordination detection problem.

\section{Multi-Objective Optimization and Revealed Preferences}
\label{sec:MOO}

In order to characterize conditions under which coordinaton can be detected by an outside observer, one much precisely define what is meant by coordination in the first place. Notions of coordination have appeared in e.g., \citep{chen2020toward}, \citep{quintero2010optimal},\citep{wise2006uav}. We utilize a well-motivated and widely used framework to define coordination, known as multi-objective optimization. In this section we present the mathematical details of multi-objective optimization and \textit{inverse} multi-objective optimization, and give a microeconomic result allowing us to acheive the latter efficiently. The application of these frameworks to the UAV coordination detection problem will be detailed in the following sections. 

Here we outline what distinguishes multi-objective optimization from single-objective optimization, and provide the resultant generalized notion of a solution concept.
\subparagraph{Multi-Objective Problem} 
We consider a system composed of multiple autonomous agents. Each agent has an individual utility function which captures their objective, and aims to act in a way that maximizes their utility function. In order to capture a notion of coordination it is assumed that there is a joint constraint on the actions taken, such that both the set of all actions which can be taken by a particular agent and the resultant utility achieved by this agent, are dependent on the actions taken by \textit{all} of the agents. This coupling forces the set of all agents to jointly consider the actions taken in order to achieve individual objectives. 

\subparagraph{Multi-Objective Solution Concept} The reader may realize that this is also the setting of game theory, where a standard investigation is that of non-cooperative agents acting solely in self-interest. The classical solution concept in non-cooperative game theory is that of Nash Equilibrium, where no agent can gain in their utility by unilaterally deviating (changing their action). We distinguish this from the \textit{cooperative} solution concept in multi-objective optimization, that of \textit{Pareto-optimality}. Pareto optimality occurs when no agent can gain in their utility by unilaterally deviating (changing their action) \textit{without simultaneously decreasing the utility of another agent}. So, an individual agent could feasibly change their action to increase their utility, but this would come at the expense of decreasing another agent's utility. Thus, a Pareto-optimal solution captures a notion of coordination, since the agents do not act in complete self-interest but act in order to maximize the entire set of utility functions. 


\subparagraph{Inverse Multi-Objective Problem}
Now that the multi-objective problem has been conceptualized, one may ask: given a dataset of actions, how can it be determined if the group is behaving in a Pareto-optimal manner? This general problem is denoted as inverse multi-objective optimization, and originated from the recovery of decision process structures in microeconomic group behavior analysis \citep{chiappori2009microeconomics}. 
More specifically, in inverse multi-objective optimization we aim to determine \textit{if there exist} individual utility functions for which the actions are multi-objective optimal. If so, we aim to reconstruct such utility functions in order to better understand or predict the system dynamics. A key framework for accomplishing this will be that of microeconomic revealed preferences. 

\subparagraph{Revealed Preferences} The micro-economics literature contains the most well-developed formulations of such inverse multi-objective optimization, nominally 'Group Revealed Preferences'. The Revealed Preferences paradigm dates back to seminal work \citep{afriat1967construction}, where utility maximization behavior is detected from consumer budget-expenditure data. The Group Revealed Preferences \citep{cherchye2011revealed} formulation extends these works to the multi-agent scenario, giving necessary and sufficient conditions for group behavior to be consistent with multi-objective optimization. Furthermore, a methodology is provided for reconstructing feasible utility functions under which the observed behavior is multi-objective optimal. This allows for inference of multi-agent group motives or prediction of future behavior. 

\paragraph{Outline} 
In the rest of this section we make the above concepts more mathematically precise: we first outline the mathematics of multi-objective optimization, then provide the relevant framework for inverse multi-objective optimization, given by the micro-economic Group Revealed Preference formulation. In the following section we utilize these mathematical tools in the UAV coordination detection problem.

\subsection{Multi-Objective Optimization}
In this section we introduce the multi-objective optimization problem we will consider, then present its solution concept of Pareto-optimality, and discuss how Pareto-optimal solutions can be obtained. 

\subparagraph{Multi-Objective Problem}

We consider $M \in \nat$ agents. We denote $\var \in \reals^n$ a general joint-action taken by all agents. E.g., $\var$ can represent a vector containing distinct actions taken by each agent, or it can represent a single action that has been agreed upon by the set of agents. Each agent $i\in[M] := \{1,\dots,M\}$ has a \textit{utility function} $f^i: \reals^n \to \reals$, representing agent $i$'s utility gained from the joint-action taken. 

This setting is sufficiently general to capture standard game-theoretic notions. For instance, in non-cooperative Game Theory, the joint-action $\var$ can represent the set of distinct actions taken by each agent. Then solution concepts such as Nash Equilibria, where no agent has an incentive to unilaterally deviate from its action, can be studied. 

Our focus in this setting will instead be on a notion of multi-agent \textit{coordination}, given by a particular linearly-constrained multi-objective optimization:

\vspace{0.5cm}

\fbox{%
\parbox{0.95\linewidth}{%
\textbf{Linearly Constrained Multi-Objective Optimization}
\begin{align}
\begin{split}
\label{eq:MOP}
 &\arg \max _\var \{f^1(\var),\dots,f^{\numagents}(\var)\} \\
 &\text{s.t. } \var \in \conset := \{\gamma \in \mathbb{R}^n  : \alpha'\gamma \leq 1 \}
\end{split}
\end{align}
}
}

\vspace{0.5cm}

\eqref{eq:MOP} encodes the idea that the agents must cooperate such that joint-action $\beta$ maximizes over all objective functions $f^i$, provided $\beta$ is in a linear constraint set $\alpha'\beta \leq 1$ formed by \textit{constraint vector} $\alpha$. The linear constraint $\alpha' \var$ \footnote{For vector $x$ we let $x'$ represent the transpose of $x$.}is bounded by 1 without losing generality (see Sec. I-A of \citep{krishnamurthy2020identifying}). 

The astute reader may at this point ask what precisely is meant by the maximization in \eqref{eq:MOP}. Indeed, it turns out that we need to introduce a generalized notion of optimality in order for this maximization to be well-posed. 

\subparagraph{Multi-Objective Solution Concept. Pareto Optimality} In single-objective optimization, the goal is to find a feasible argument which maximizes the objective, in that the objective evaluated at this argument is greater than or equal to the objective evaluated at any other point in the feasible set. A naive generalization of this to the multi-objective setting might be to find an argument which maximizes all objectives. However, unless there are very tight restrictions on the objective function structures (e.g., all the same function, or all one-dimensional and monotone) there will seldom exist an argument $\var$ which simultaneously maximizes all objectives. Thus, there will be tradeoffs between objectives for varying argument $\var$. The general solution concept for the multi-objective optimization problem \eqref{eq:MOP} that captures these tradeoffs is instead that of \textit{Pareto optimality}:

\begin{definition}\textbf{Pareto Optimality}:\\
\label{def:par_opt}
 \normalsize For fixed $\{\{f^i(\cdot)\}_{i=1}^{\numagents},\alpha \}$ and a vector $\var \in \conset = \{\gamma \in\reals^n: \,\alpha'\gamma\leq 1\}$, let 
\begin{align*}
\begin{split}
    &Z^t(\var) = \{\gamma \in \conset : f^i(\gamma) \geq f^i(\var) \ \forall i \in [\numagents]\} \\
    &Y^t(\var) = \{\gamma \in \conset : \exists k\in[M] : f^k(\gamma) > f^k(\var) \}
\end{split}
\end{align*}
The vector $\var$ is said to be \textit{Pareto-optimal} if 
\begin{equation}
\label{efficiency}
Z^t(\var) \cap Y^t(\var) = \emptyset
\end{equation}
where $\emptyset$ denotes the empty set.
\end{definition}

In words, a vector $\var$ is Pareto-optimal if there does not exist another vector $\gamma$ in the feasible set $\conset$ which increases the value of some objective $f^i(\cdot)$ without simultaneously decreasing the value of some other objective $f^j(\cdot)$, $i,j\in [\numagents]$. 

This is a well-motivated and nontrivial conception of cooperative optimality in multi-agent systems \citep{marden2014achieving}, \citep{ruadulescu2020multi}. It captures the idea that even if a single agent may gain by deviating from the Pareto-optimal joint-action, it does not do so since that gain would come at the expense of another agent. From another perspective, if a joint-action is not yet Pareto-optimal it means that it can be altered such that no agents' utility decreases and at least one agent's utility increases. Such an alteration may have to be undertaken by a certain agent who gains nothing by changing their action, but does so in order to increase the utility of a different agent. Thus, achieving the Pareto-optimum conceptually corresponds to all agents simultaneously acting for the best of the entire group.

In general there will be a set of Pareto-optimal solutions, some benefiting certain individual agents more than others, but all maximizing the utilities of the entire group in the above described sense.

\begin{definition}\textbf{Pareto Frontier}: \\
The set of \textit{all} Pareto-optimal solutions to the problem \eqref{eq:MOP} is known as the \textit{Pareto-frontier}, and is denoted 
\begin{equation}
\label{eq:effset}
\effset(\{f^i\}_{i=1}^M, \lconst) := \{\var \in \conset : \eqref{efficiency} \textrm{ is satisfied}\}
\end{equation}
\end{definition}

Now, we say that $\var$ solves \eqref{eq:MOP} if and only if $\var$ is Pareto-optimal, i.e.
\begin{align*}
   &\var \in \{ \arg \max_\var \{f^1(\var),\dots,f^{\numagents}(\var)\}\ s.t. \ \var \in \conset \}\ \Longleftrightarrow \ \var \in \effset(\{f^i\}_{i=1}^M, \alpha)
\end{align*}

\subparagraph{Computing Pareto Optimal Solutions}

We have discussed the multi-objective optimization problem and its solution concept of Pareto-optimality. The question remains: given joint-action constraints and individual utility functions, how can one (or the multi-agent group itself) actually compute Pareto-optimal solutions? Here we show how Pareto-optimal solutions can be obtained by simply maximizing linear combinations of objective functions $\{f^i\}_{i=1}^M$ subject to the linear constraint $\alpha'\var \leq 1$.

Before presenting this result we need to introduce some notation. Let $\sw = (\sw^1,\dots,\sw^{\numagents})' \in \reals^{\numagents}_{\geq0}$ be a set of real-valued weights on the non-negative unit simplex $\simplex$, defined as 
\begin{equation}
    \simplex := \{\sw \in \reals^{\numagents}_{\geq0} : \boldsymbol{1}'\sw = 1\}.
\end{equation} Also let  
\begin{equation}
\psimplex := \{\sw \in \reals^{\numagents}_{+} : \boldsymbol{1}'\sw = 1\} \subset \simplex
\end{equation} be the set of strictly positive weights. Let us denote 
\[\woptset := \biggl\{\beta : \,\beta \in \arg \max_{\gamma}\sum_{i=1}^M \sw^i f^i(\gamma) \, s.t.\, \alpha'\gamma \leq 1 \biggr\}\]
i.e., $\woptset$ is the set of all vectors $\beta$ maximizing a linear combination of objective functions $\{f^i\}$ with weights $\mu$, such that the linear constraint $\alpha' \beta \leq 1$ is satisfied. Then, we have following relation \citep{miettinen2012nonlinear}: 

\begin{equation}
\label{eq:effrel}
\bigcup_{\sw \in \psimplex}\woptset \subseteq \effset(\{f^i\}_{i=1}^M, \lconst) \subseteq \bigcup_{\sw \in \simplex} \woptset
\end{equation}
where the second inclusion is an equality if the objective functions are concave. Relation \eqref{eq:effrel} implies that if we solve 
\begin{equation}
\label{eq:linmax}
    \arg\max_{\gamma}\sum_{i=1}^M\mu^i f^i(\gamma) \, s.t.\, \alpha'\gamma \leq 1 
\end{equation} with weights $\mu$ strictly positive, then this solution is guaranteed to be Pareto-optimal. Furthermore, provided the objective functions $f^i$ are concave, \textit{all} Pareto-optimal solutions can be produced by solving \eqref{eq:linmax} with weights varying over the non-negative simplex $\simplex$. In particular, this is useful since \eqref{eq:linmax} is a \textit{constrained single-objective optimization}, which can be computed efficiently in most cases if the utility functions are concave.  

\subsection{Inverse Multi-Objective Optimization}
\label{sec:imoo}
In this section we make the concept of inverse multi-objective optimization mathematically precise, and introduce a key theorem enabling us to achieve it in a general microeconomic framework. 
\subparagraph{Inverse Multi-Objective Problem}
The inverse multi-objective optimization problem can be stated conceptually as follows. Given constrained outputs (actions) of an observed multi-agent system, does there exist a set of utility functions under which the observed outputs are multi-objective optimal? Can these utility functions be reconstructed? At first, a mathematical instantiation of this statement might be: Given $(\alpha,\beta)$, does there exist a set of utility functions $\{f^i\}_{i=1}^M$ and weights $\mu \in \simplex$ such that 
\begin{equation}
\label{eq:bargmax}
\beta \in \arg\max_{\gamma}\sum_{i=1}^M \sw^i f^i(\gamma) \, s.t. \, \alpha'\gamma \leq 1
\end{equation}
If there exist such a set of weights $\mu$ and utility functions $\{f^i\}_{i=1}^M$ then we say that the data $(\alpha,\beta)$ is \textit{rationalized} by these weights and utility functions. However, for a single data-point $(\alpha,\var)$, there will \textit{always} exist sets $\{f^i\}$ and $\mu$ which rationalize it. To see this, take $\mu$ in the corner of the simplex, such that $\mu^i=1$ for some $i$ and $\mu^j=0 \ \forall j\neq 1$. Then \eqref{eq:bargmax} reduces to $\var \in \arg\max_{\gamma} f^i(\gamma) \, s.t. \, \alpha' \gamma \leq 1$, and obviously one can find some $f^i$ for which this is true. Thus, the inverse multi-objective optimization with a single data-point is trivial.

To make the problem non-trivial, we consider multiple data-points indexed by time, i.e., suppose we observe the dataset $\dataset := \{\lconst, \beta_t\}$
of constraint vectors $\lconst$ and system outputs $\beta_t$ indexed over discrete-time $t\in[T] := \{1,\dots,T\}$. Then this extended inverse multi-objective optimization problem can be stated as follows:

\vspace{0.5cm}

\fbox{%
\parbox{0.95\linewidth}{%
\normalsize
\textbf{Inverse Multi-Objective Optimization} \\
Given a time-indexed dataset $\dataset := \{\lconst, \beta_t\}$, do there exist utility functions $\{f^i\}_{i=1}^M$ such that 
\[\beta_t \in \arg\max_{\gamma} \sum_{i=1}^M \sw^i f^i(\gamma) \, \, s.t. \, \, \lconst' \gamma \leq 1 \quad \forall t \in [T]\]
for some weights $\sw$ in the simplex $\simplex$? If so, how can one reconstruct these utility functions? 
}}

\vspace{0.5cm}

The above problem is distinct from the (trivial) single data-point problem explained above, since here the utility functions $\{f^i\}_{i=1}^M$ must rationalize the data $\{\lconst,\beta_t\}$ \textit{for all} $t\in[T]$ simultaneously. One can easily see how this distinction makes the problem non-trivial, since the set of utility functions which rationalize the data-set for some fixed time-point may not rationalize the data for another time-point. In this sense, the inverse multi-objective optimization problem tests whether a multi-agent system behaves optimally (in the Pareto-sense) at each time point, and is also consistent in behaving optimally (w.r.t. the same utility functions) over all tested time-points. 

Figure~\ref{fig:blockdiagram} provides an illustration of the procedure for inverse multi-objective optimization, in relation to the generative process of multi-objective optimization.

Next we discuss a micro-economic solution to a specific form of this problem.

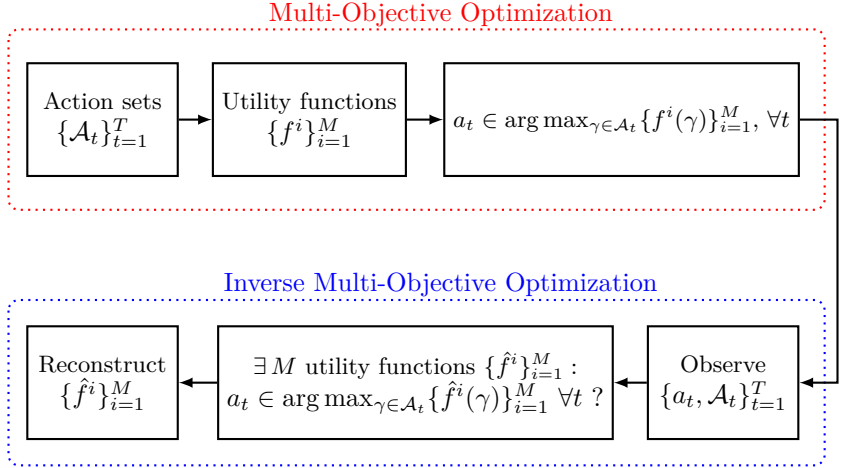
\begin{figure}
\centering

\begin{tikzpicture}[>=latex, thick]

\draw[dotted, thick, rounded corners, red] (-1.25,1.2) rectangle (9.6,-1.2);

\node[red] at (4.5,1.4) {Multi-Objective Optimization};

\draw[dotted, thick, rounded corners, blue] (-1.25,-2.4) rectangle (9.6,-4.6);

\node[blue] at (4.5,-2.2) {Inverse Multi-Objective Optimization};

\node[draw, rectangle, minimum width=2cm, minimum height=1.5cm, align=center] (action) at (0,0) {\small Action sets\\ $\{\mathcal{A}_t\}_{t=1}^T$};
\node[draw, rectangle, minimum width=2cm, minimum height=1.5cm, align=center] (utility) at (2.75,0) {\small Utility functions\\ $\{f^i\}_{i=1}^M$};
\node[draw, rectangle, minimum width=4cm, minimum height=1.5cm, align=center] (optimal) at (6.9,0) {\small $a_t \in \arg\max_{\gamma\in\mathcal{A}_t} \{f^i(\gamma)\}_{i=1}^M,\, \forall t$};

\draw[->] (action) -- (utility);
\draw[->] (utility) -- (optimal);

\node[draw, rectangle, minimum width=2cm, minimum height=1.5cm, align=center] (result1) at (0,-3.5) {\small Reconstruct \\ $\{\hat{f}^i\}_{i=1}^M$};
\node[draw, rectangle, minimum width=4.5cm, minimum height=1.5cm, align=center] (result2) at (4.16,-3.5) {\small $\exists \, M $ utility functions $\{\hat{f}^i\}_{i=1}^M : $ \\$a_t \in \arg\max_{\gamma\in\mathcal{A}_t} \{\hat{f}^i(\gamma)\}_{i=1}^M \,\, \forall t$ ?};
\node[draw, rectangle, minimum width=2cm, minimum height=1.5cm, align=center] (result3) at (8.25,-3.5) {\small Observe \\$\{a_t,  \mathcal{A}_t\}_{t=1}^T$};

\draw[->] (result2) -- (result1);

\draw[->] (result3) -- (result2);

\draw[->] (optimal.east) -- +(0.5,0) |-  (result3.east);

\end{tikzpicture}
\caption{Forward and inverse multi-objective optimization. The (forward) multi-objective optimization problem consists of a set of feasible actions and a utility function for each agent. The optimization problem is to find an action $a_t$ that maximizes over the set of utility functions. The \textit{inverse} multi-objective optimization problem is to observe the actions taken, and first determine if there exist individual utility functions making the actions Pareto-optimal. Then, if so, these 'rationalizing' utility functions should be reconstructed.}
\label{fig:blockdiagram}
\end{figure}

\subparagraph{Group Revealed Preferences}
The microeconomic field of Revealed Preferences aims to detect utility maximization behavior among observed consumers. We present here the form of multi-objective optimization considered in this literature, which is a special case of the general multi-objective problem \eqref{eq:MOP}. 
Suppose we have the dataset of constraints and system responses $\dataset = \{\lconst, \{\respi\}_{i=1}^{\numagents}, t \in [T] \}$. Here $\beta_t^i$ corresponds to the action taken by agent $i$. We say the dataset satisfies "collective rationality" if it solves the following multi-objective optimization problem:

\vspace{0.5cm}

\fbox{%
\parbox{0.95\linewidth}{%
\normalsize \textbf{Microeconomic Collective Rationality}\\
$\exists \, \mu \in \simplex, \{U^i\}_{i=1}^M$, $U^i: \reals^{\dimn} \to \reals$ concave and monotone increasing such that: 
\begin{equation}
\label{eq:MCR}
    \{\respi\}_{i=1}^M \in \arg\max_{\{\gamma^i\}_{i=1}^M} \sum_{i=1}^{\numagents} \mu^i U^i(\gamma^i) \quad \textrm{ s.t. } \lconst'\left( \sum_{i=1}^{\numagents} \respi\right)\leq 1 \quad \forall\,t
\end{equation}
}}

\vspace{0.5cm}
Notice that This form of "collective rationality" can be obtained as a special case of the more general form \eqref{eq:MOP}, where each agent's utility function is only explicitly dependent on it's own action. However, \eqref{eq:MCR} still optimizes over joint-actions in the same sense as \eqref{eq:MOP} since the linear constraint limits the sum of individual actions. 

The inverse multi-objective problem in this specialized case then is analogous to the general problem in the previous subsection: we ask if there exist utility functions such that \eqref{eq:MCR} holds for all $t$. In \citep{cherchye2011revealed}, a necessary and sufficient condition is derived for the dataset $\dataset$ to be consistent with this notion of multi-objective optimization.

\begin{theorem}
 \label{thm:cherchye1}
    Let $\dataset = \{\lconst, \{\respi\}_{i=1}^{\numagents}, t\in[T]\}$ be a set of observations. The following are equivalent:
    \begin{enumerate}
    \item there exist a set of $M$ concave and continuous objective functions $U^1,\dots,U^m$, weights $\sw \in \psimplex$ and constraint $p^*$ such that $\forall t \in [T]$:
    \begin{align}
    \begin{split}
    \label{thm1:rat}
        \{\respi\}_{i=1}^{\numagents} \in &\arg\max_{\{\argo^i\}_{i=1}^{\numagents}} \sum_{i=1}^{\numagents} \sw^i U^i(\argo^i) \ \  s.t. \ \lconst' (\sum_{i=1}^{\numagents}\argo^i ) \leq p^*
    \end{split}
    \end{align}
    \item there exist numbers $u_j^i \in \reals, \lambda_j^i > 0$ such that for all $s,t \in [T]$, $i \in [M]$: 
    \begin{equation}
    \label{af_ineq}
        u_s^i - u_t^i - \lambda_t^i\lconst'[\var_s^i - \varti] \leq 0
    \end{equation}
    \end{enumerate}
\end{theorem}
\begin{proof}
See Proposition 1 of \citep{cherchye2011revealed}
\end{proof}

Furthermore, if the above conditions hold, then specific utility functions which "rationalize" the dataset can be reconstructed in the following way. 

\begin{corollary}
\label{cor:Utrec}
Given constants $u_t^i, \lambda_t^i, t\in[T],i\in[M]$ which make \eqref{af_ineq} feasible, explicit monotone and continuous objective functions that "rationalize" the dataset \\ $\{\lconst, \respi, t \in [T], i \in [M]\}$ are given by
\begin{equation}
\label{eq:Utrec}
     U^i(\cdot) = \min_{t \in [T]} \left[u_t^i + \lambda_t^i\lconst'[\cdot - \respi] \right]
\end{equation}
i.e., \eqref{thm1:rat} is satisfied with objective functions \eqref{eq:Utrec}.
\end{corollary}
\begin{proof}
See Lemma 1 of \citep{snow2022identifying}.
\end{proof}

These results give us a principled and efficient way of performing inverse multi-objective optimization, by testing the feasibility of a linear program. We can first test whether the data is consistent with "collective rationality", i.e., whether the group is behaving "intelligently" by consistently optimizing a set of utility functions, then we can reconstruct individual utility functions which rationalize the dataset. This gives us a mechanism for inferring the underlying distribution of objectives in the group, or for predicting future group behavior. 

In this section we have first presented the general forward and inverse multi-objective optimization problems, then revealed a specific form of multi-objective optimization that can be tested efficiently by solving a particular linear program. Next, we present the setting which we will apply these results: detecting UAV coordination. We first outline the UAV-tracking dynamics and interaction model, then show how this can be mapped to the setting presented in this first Section, allowing for efficient testing of UAV "coordination".

\section{Multi-Objective Optimization in UAV Networks}
\label{sec:UAV}

In this section we consider the specific instantiation of a radar - UAV network tracking scenario. The multi-objective optimization framework presented in the previous section will allow us to precisely define coordination in the UAV network, and efficiently detect such coordination on the radar's end. In this section we
\begin{itemize}
    \item present the radar - UAV network interaction dynamics,
    \item provide the definition of UAV network coordination,
    \item outline several motivational target-tracking frameworks which give rise to the above notion,
\end{itemize}

\begin{figure}
\centering
  \includegraphics[width=0.5\linewidth,scale=0.25]{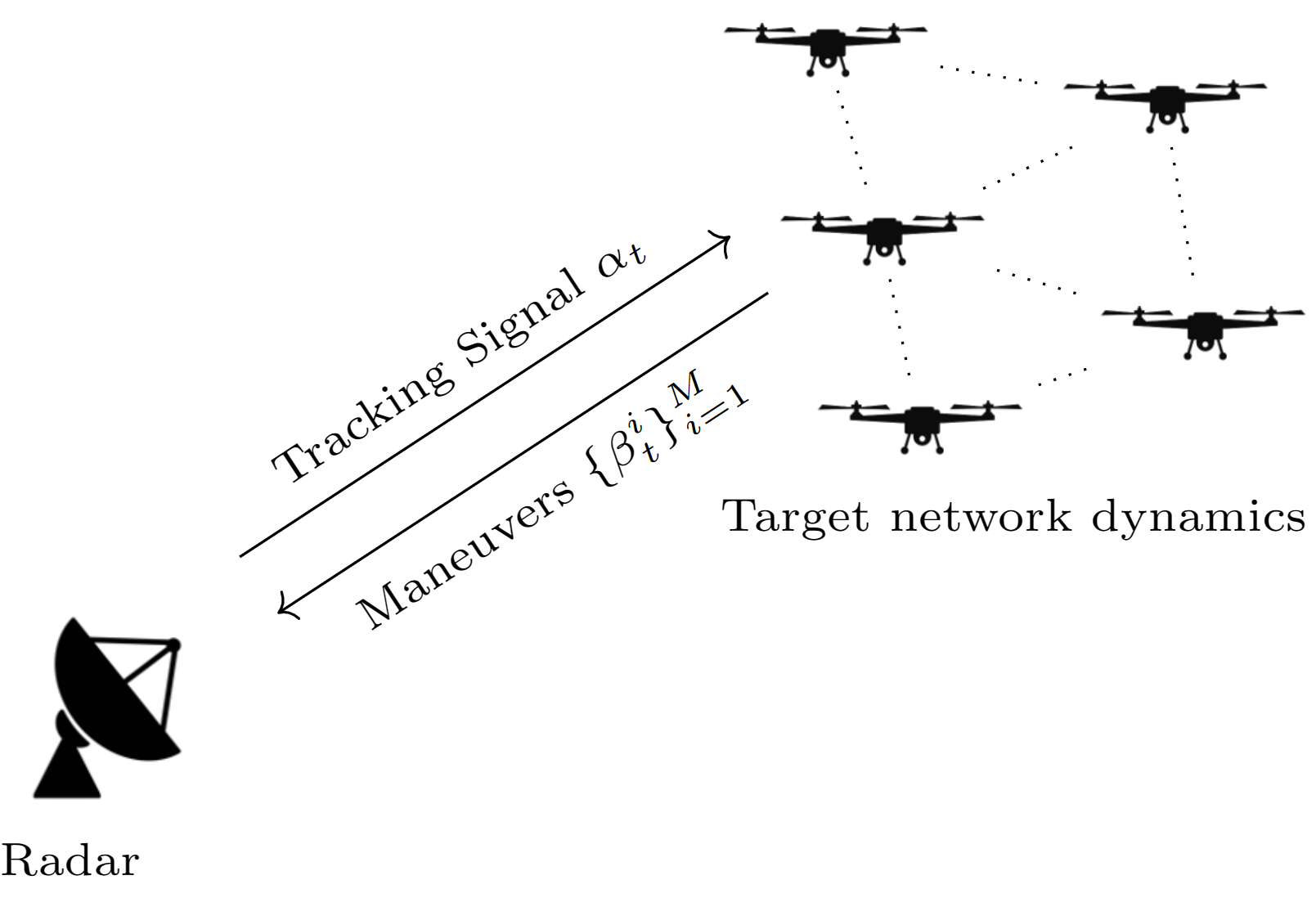}
  \caption{\small UAV Network Interaction. We represent the high-level radar tracking waveform (parameters) by $\alpha_t$, and the target network menauevers by $\{\beta_t^i\}_{i\in[M]}$.}
  \label{fig:interaction}
\end{figure}

\subsection{Interaction Dynamics}
Here we provide the general interaction dynamics between a UAV (target) network and a radar (us). For now let us define, at time $t\in\nat$, the radar's tracking signal as $\alpha_t$ and target $i$'s maneuver as $\beta_t^i$. Figure~\ref{fig:interaction} displays the high-level interaction dynamics: The radar probes the target network, and obtains measurements of the network maneuvers. We will momentarily give explicit motivation for how these variables can be interpreted in a physical-layer multi-target tracking scenario. We consider inverse multi-objective optimization; we aim to detect whether the target network coordinates in a specific sense (corresponding to our previous notion of multi-objective optimization).

At an implementation level, we aim to detect whether the targets jointly adjust their maneuvers such that their overall utility is maximized (in the Pareto-optimal sense), subject to a constraint on their \textit{detectability} by the radar. We will also momentarily provide a definition and motivation for such a notion of detectability.

We consider two time scales for the interaction: the fast time scale $\ftime = 1,2,\dots$ represents the scale at which the target state and measurement dynamics occur, and the slow time scale $\stime = 1,2,\dots$ represents the scale at which the radar probes (tracking signals) and UAV maneuvers $\{\varti\}_{i=1}^{\numagents}$ occur.  

\begin{definition}[Radar - Multi-Target Interaction]
The radar - UAV network interaction has the following dynamics:
\begin{align}
\label{inter_dynam}
    \begin{split}
        \textrm{radar emission}: \lconst &\in \reals^{\dimn}_+ \\
        \textrm{UAV i maneuver} : \varti &\in \reals^{\dimn}_+ \\
        \textrm{UAV i state} : x_{\ftime}^i &\in \reals^{\sdim}, \
        x_{\ftime + 1}^i \sim p_{\varti}(x| x_{\ftime}^i) \\
        \textrm{radar observation}: \measi_{\ftime} &\in \reals^{\mdim}, \
        \measi_{\ftime} \sim p_{\lconst}(y|x_{\ftime}^i) \\
        \textrm{radar tracker}: \pi^i_{\ftime} &= \mathcal{T}(\pi^i_{\ftime-1},\measi_{\ftime}) 
    \end{split}
\end{align}
\end{definition}
where $\pi_k^i$ is radar $i$'s target state posterior and $\mathcal{T}$ is a general Bayesian tracker. For a fixed $\stime$ in the slow time-scale, $\lconst$ abstractly represents the radar's signal output which parameterizes its measurement kernel, and $\varti$ represents target $i$'s maneuver (radial acceleration, etc.) which parameterizes the state update kernel. These interaction dynamics are illustrated in Fig.~\ref{fig:intDyn}. 
Taking the point of view of the radar, we aim to detect if the targets are \textit{coordinating}.

We next present precisely what is meant by coordination, and motivate how the mathematical definition can be derived from practical multi-target filtering algorithms. 

\begin{figure}
\centering
  \includegraphics[width=0.6\linewidth,scale=0.25]{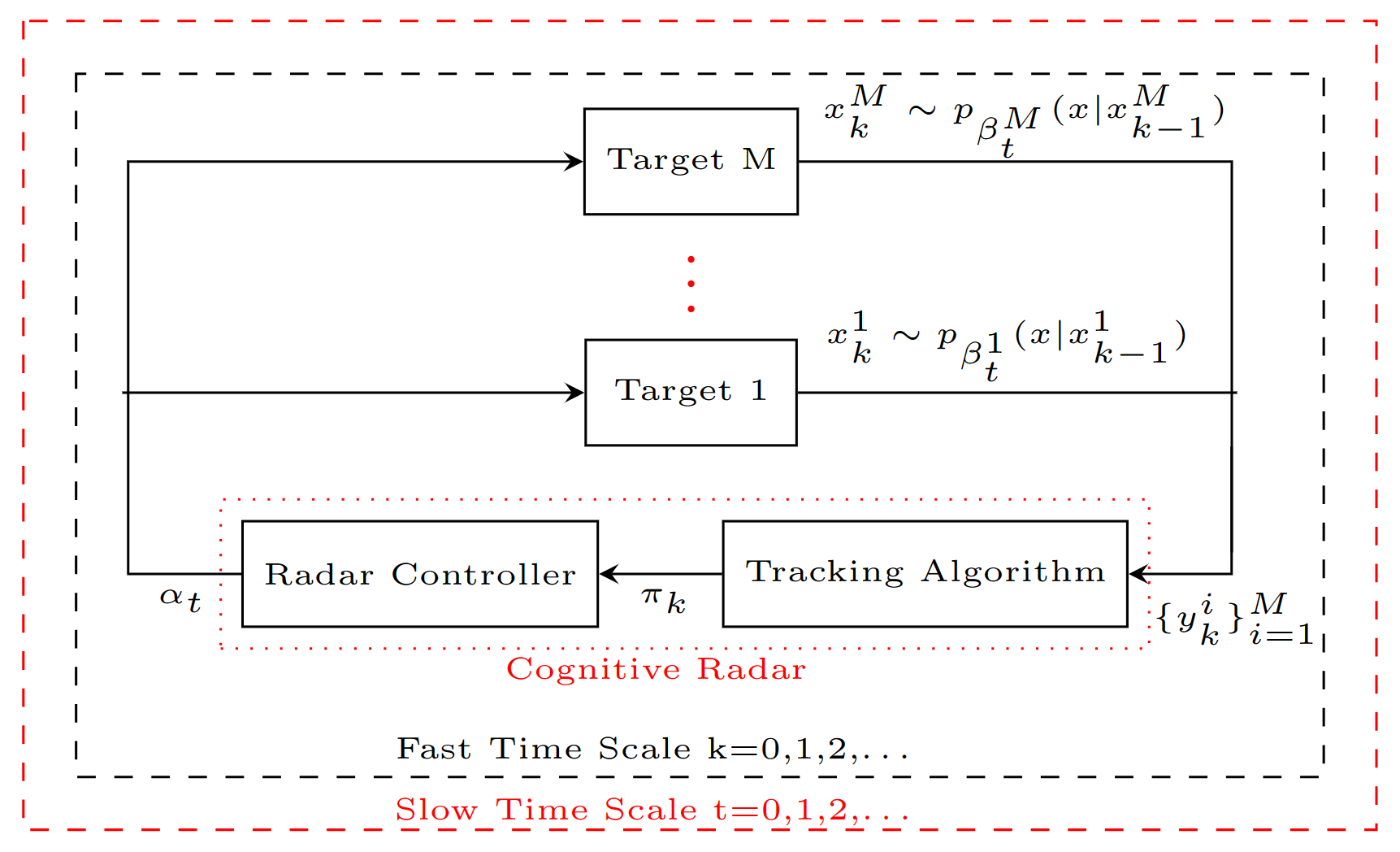}
  \caption{\small UAV Network Interaction Dynamics. The interaction occurs at two time-scales. The slow time-scale, indexed by $t\in\nat$, is the scale at which radar waveform signal parameters $\alpha_t$ and target maneuvers $\{\beta_t^i\}_{i=1}^M$ are adjusted. For a fixed radar tracking waveform and set of target maneuvers, the radar obtains a sequence of target measurements $\{y_k^i\}_{i\in[M}$, indexed on the fast time-scale by $k\in\nat$. From these measurements the radar implements a multi-target filtering algorithm to track the states $\{x_k^i\}_{i\in[M]}$, and thus can recover $\{\beta_t^i\}$. }
  \label{fig:intDyn}
\end{figure}

\subsection{UAV Network Coordination: Constrained Spectral Optimization}
\label{sec:cso}
Here we present a correspondence between the spectral UAV network dynamics and a constrained multi-objective optimization problem, thereby defining what is meant by coordination and showing how it arises from the interaction dynamics \eqref{inter_dynam}.

\paragraph{UAV Network Coordination}
In formulating our problem, it is necessary to define rigorously what we mean by UAV coordination. Examples of such coordination definitions have been proposed and studied in works \cite{snow2022identifying}, \cite{snow2023statistical}, \cite{shi2017power}. We consider the following coordination specification. Each UAV has an individual utility function $f^i$, which maps from its state dynamics $\beta_t^i$, parametrizing the state transition kernel in \eqref{inter_dynam}, to a real-valued utility, i.e., 
\[f^i: \reals^{\dimn} \to \reals\]
Such utility functions can capture the UAVs' flight objectives by quantifying a reward profile for flight dynamics. The UAVs then should act to maximize their individual utility functions at each point in time in order to achieve their flight objective. However, such individual maximization would decouple the UAV dynamics such that they act independently of each other's trajectories. A notion of coordination would need to capture a certain coupling or codependency between these trajectories.

We propose to quantify this coupling through a constraint on the radar's average measurement precision. This captures the idea that the UAVs aim to obtain some flight objective while jointly acting such that the entire network remains hidden to a certain degree from the radar. This induces a coupling between UAV trajectories; the UAVs must adjust their individual sequential state dynamics such that the entire network satisfies a certain undetectability constraint. 

This coordination formulation can be summarized informally as:
\begin{align}
    \begin{split}
    \label{moo_int}
        &\textrm{maximize }(f^1,\dots,f^{\numagents})\textrm{,  such that} \\ &\textrm{average radar measurement precision} \leq \textrm{bound}
    \end{split}
\end{align}

The 'maximize $(f^1,\dots,f^{\numagents})$' can be interpreted in the framework of Pareto optimality, as introduced in the previous section. The radar measurement precision bound can be derived from standard multi-target tracking algorithms, as we show in the following section. 

This leads us to our formal definition of coordination in a UAV network, given as follows:
\begin{definition}[Coordinating UAV Network]
\label{def:coord}
Considering the interaction dynamics \eqref{inter_dynam}, we define a coordinating UAV network to be a network of $\numagents$ UAVs, each with individual concave, continuous and monotone increasing\footnote{This objective function structure is known as 'locally non-satiated' in the micro-economics literature, and is not necessarily restrictive when considering target objectives, see \citep{krishnamurthy2020identifying}.} objective functions $f^i: \reals^{\dimn} \to \reals, i\in[\numagents]$, which produces output signals $\{\varti\}_{i=1}^{\numagents}$ on the slow time-scale in accordance with
\begin{align}
\label{def:coord_eq}
    \begin{split}
        \{\varti\}_{i=1}^{\numagents} \in &\arg\max_{\{\argo^i\}_{i=1}^{\numagents}} \{f^1(\argo^1),\dots,f^{\numagents}(\argo^{\numagents})\} \\
        & s.t. \quad \lconst' (\sum_{i=1}^{\numagents} \argo^i) \leq 1 
    \end{split}
\end{align}
    
\end{definition}

Note that \eqref{def:coord_eq} is a special case of the general multi-objective optimization problem \eqref{eq:MOP}, in which the objective functions do not share a common argument but the arguments are jointly constrained. Thus, a coordinating UAV network controls its joint state dynamics (through e.g., controlling a certain formation) such that they are \textit{Pareto optimal} (Def. \eqref{def:par_opt}) with respect to each objective function, the tracking signal from the radar, and a constraint on the UAV network's detectability. 

It is quite straightforward to interpret the individual utility functions $f^i$ of the targets as encoding flight objectives, but one may well ask how the linear constraint in \eqref{def:coord_eq} corresponds to a bound on the radar's average measurement precision, as suggested in the informal definition \eqref{moo_int}. We next provide an example of multi-target state dynamics and several resultant radar tracking algorithms which naturally give rise to this constraint. The purpose is to shed light on how the abstract constrained multi-objective optimization \eqref{def:coord_eq} can be recovered from practical filtering-level tracking dynamics.

\paragraph{Multi-Target Spectral Dynamics}
Here we specify a concrete example of the abstract dynamics \eqref{inter_dynam}. Linear Gaussian dynamics for a target's kinematics \citep{li2003survey} and linear Gaussian measurements at each radar are widely assumed as a useful approximation \citep{bar2004estimation}. Thus, we will consider the following linear Gaussian state dynamics and measurements over the \textit{fast time scale $\ftime \in \nat$}, with a particular $t\in\nat$ fixed:
\begin{align}
    \begin{split}
    \label{lin_gaus}
        \statei_{\ftime+1}^i &= A^i\statei_{\ftime}^i + \snoisei_{\ftime}^i, \  \statei_0^i \sim \pi^i_0, \\
        \measi_{\ftime} &= C^i\statei_{\ftime}^i + \mnoisei_{\ftime}, \ i\in[\numadv]
    \end{split}
\end{align}
where $\statei_{\ftime}^i, \snoisei_{\ftime}^i \in \reals^{\sdim}$ are the target $i$ state and noise vectors, respectively, and $A^i \in \reals^{\sdim\times \sdim}$ is the state update matrix for target $i$. $\measi_k \in \reals^{\mdim}$ is the radar's measurement of target $i$, $C^i \in \reals^{\mdim \times \sdim}$ is the measurement transformation, and $\mnoisei_k \in \reals^{\mdim}$ is the measurement noise. The constraints and subsequent radar responses will be indexed over the \textit{slow time scale} $\stime \in \nat$. Abstractly, these will parameterize the state and noise covariance matrices:
\begin{equation}
\label{eq:noise}
    \snoisei_{\ftime} \sim \gaussN(0,Q_t(\beta_t^i)), \ \mnoisei_{\ftime} \sim \gaussN(0,R_t(\alpha_t))
\end{equation}

 In this spectral interpretation, $\varti$ represents the vector of eigenvalues of state-noise covariance matrix $Q_t$ and $\lconst$ represents the vector of eigenvalues of the inverse measurement noise covariance matrix $R_t^{-1}$. Thus, given this interpretation we can view modulations of $\alpha_t$ and $\beta_t^i$ as corresponding to increased/decreased measurement precision on the part of the radar. This will be made precise subsequently when we discuss filtering details. First, we briefly illustrate how such noise covariance matrices can be parametrized in the first place. 

 \subparagraph{Waveform Design for Measurement Covariance Modulation}

 To give a precise structure to the radar dynamics, this section provides examples of how the observation noise covariance $R_t(\alpha_t)$ in \eqref{eq:noise} can depend on the radar waveform. Further details on maximum likelihood estimation involving the radar ambiguity function can be found in \citep{van2004detection}, \citep{kershaw1994optimal}. The waveform specifications involve the following terms:
 \begin{itemize}
     \item $c$ denotes the speed of light (in free space),
     \item $\omega_c$ denotes the carrier frequency,
     \item $\theta$ is an adjustable parameter in the waveform,
     \item $\eta$ is the signal to noise ratio at the radar,
     \item $j = \sqrt{-1}$ is the unit imaginary number,
     \item $s(t)$ is the complex envelope of the waveform,
     \item $\alpha$ is the vector of eigenvalues of $R^{-1}$
 \end{itemize}
 We now provide three example waveforms and their resulting observation noise covariance matrices $R(\alpha)$:
 \begin{enumerate}
\item Triangular Pulse - Continuous Wave
\begin{align*}
\begin{split}
    s(t) &= \begin{cases}
        \sqrt{\frac{3}{2\theta}}\left(1 - \frac{|t|}{\theta} \right) \quad &-\theta < t < \theta \\
        0 &\textrm{otherwise} 
    \end{cases} \\
    R(\alpha) &= \begin{bmatrix}
    \frac{c^2\theta^2}{12\eta} & 0 \\
    0 & \frac{5c^2}{2\omega_c^2\theta^2\eta}
    \end{bmatrix}
\end{split}
\end{align*}

\item Gaussian Pulse - Continuous Wave
\begin{align*}
\begin{split}
    s(t) &= \left(\frac{1}{\pi\theta^2} \right)^{1/4}\exp\left(\frac{-t^2}{2\theta^2} \right) \\
    R(\alpha) &= \begin{bmatrix}
            \frac{c^2 \theta^2}{s\eta} & 0 \\
            0 & \frac{c^2}{2\omega_c^2\theta^2\eta}
    \end{bmatrix}
\end{split}
\end{align*}

\item Gaussian Pulse - Linear Frequency Modulation Chirp
\begin{align*}
s(t) &= \left(\frac{1}{\pi \theta_1^2} \right)^{1/4}\exp\left(-\left(\frac{1}{2\theta_1^2}- j\theta_2 \right) t^2\right) \\R(\alpha) &= \begin{bmatrix}
\frac{c^2\theta_1^2}{2\eta} & \frac{-c^2\theta_2\theta_1^2}{\omega_c\eta} \\
\frac{-c^2\theta_2\theta_1^2}{\omega_c \eta} & \frac{c^2}{\omega_c^2\eta}\left(\frac{1}{2\theta_1^2} + 2\theta_2^2\theta_1^2 \right)
\end{bmatrix}
\end{align*}
\end{enumerate}
The key idea is that by adapting the waveform parameters, the radar can modulate the covariance matrix $R(\alpha)$. This modulation can be viewed at a higher level as an adaptation of the eigenvalues of $R(\alpha)$. We treat $\alpha$ as the vector of eigenvalues of $R^{-1}(\alpha)$, so that increasing $\alpha$ increases the measurement precision. Such an increase directly corresponds to, or is enacted by, changes to the physical-layer waveform parametrization, as illustrated above. 

Next, given the above Linear Gaussian specification of the multi-target dynamics \eqref{lin_gaus}, we present two multi-target filtering examples. The goal is to illustrate how the spectral interpretation of $\alpha_t$ and 
 $\beta_t^i$ in \eqref{eq:noise} gives rise within these algorithms to the linear constraint $\alpha_t(\sum_{i=1}^M\beta_t^i) \leq 1$ in \eqref{def:coord_eq}. Recall that this linear constraint should correspond to a physical-layer bound on the radar's average measurement precision.

\subsection{Multi-Target Filtering}
\label{sec:mtf}
The goal of this section is to present several multi-target tracking schemes, a simple de-coupled Kalman filter and a more complex joint probabilistic data association filter (JPDAF), and show how the high-level coordination framework \eqref{eq:linc_rec} can be recovered from each. \textit{These serve as illustrative examples of how to map complex multi-target tracking algorithms to the constrained multi-objective optimization \eqref{def:coord_eq}}. One should be able to extend these mappings to other target tracking schemes.

\paragraph{De-Coupled Kalman Filtering}
A simple interpretation of the multi-target tracking procedure is a standard de-coupled Kalman filter, whereby after measurements are associated to each target, a standard Kalman filter is applied to track each target state separately. This procedure is idealized, but allows for a nice exposition of the connection between filtering precision and the constraint in \eqref{eq:MCR}. 

\subparagraph{Filter Dynamics}
Consider the linear Gaussian dynamics \eqref{lin_gaus}, \eqref{eq:noise}. 
Based on observations $\measi_1,\dots,\measi_{\ftime}$ associated to target $i$, the tracking functionality in the radar computes the target $i$ state posterior
\begin{equation*}
\label{eq:kalpost}
    \pi_{\ftime}^i = \gaussN(\kstate_{\ftime}^i,\kcov_{\ftime}^i)
\end{equation*}
where $\kstate_{\ftime}^i$ is the conditional mean state estimate and $\kcov_{\ftime}^i$ is the covariance, computed by the classical Kalman filter:
\begin{align*}
    \begin{split}
        \kcov_{\ftime + 1 | \ftime}^i &= A^i\kcov_{\ftime}^i(A^i)' + Q_t(\varti) \\
        K_{\ftime + 1}^i &= C^i\kcov_{\ftime+1 | \ftime}^i (C^i)' + R_t(\lconst) \\
        \kstate_{\ftime+1}^i &= A^i\kstate^i + \kcov_{\ftime+1|\ftime}^i(C^i)'(K^i_{\ftime+1|\ftime})^{-1}(\measi_{\ftime+1} - C^iA^i\kstate^i_{\ftime}) \\
        \kcov^i_{\ftime+1} &= \kcov^i_{\ftime+1|\ftime} - \kcov^i_{\ftime+1|\ftime}(C^i)'(K_{\ftime+1}^i )^{-1}C^i\kcov^i_{\ftime+1|\ftime}
    \end{split}
\end{align*}
Under the assumption that the model parameters in \eqref{lin_gaus} satisfy $[A^i,C^i]$ is detectable and $[A^i,\sqrt{\sncov}]$ is stabilizable, the asymptotic predicted covariance $\kcov^i_{\ftime+1 | \ftime}$ as $k \to \infty$ is the unique non-negative definite solution of the \textit{algebraic Riccatti equation} (ARE): 
\begin{align}
    \begin{split}
    \label{eq:ARE}
        &\ARE := \\
        &- \kcov + A^i(\kcov - \kcov (C^i)'[C^i\kcov (C^i)' + R_t(\lconst)]^{-1}C^i\kcov)(A^i)' + Q_t(\varti) = 0
    \end{split}
\end{align}
  Let $\kcov_{\stime}^{*}(\lconst,\varti)$ denote the solution of the ARE and $\kcov_{\stime}^{* -1}(\lconst,\varti)$ be its inverse, representing the asymptotic measurement \textit{precision} obtained by the radar.
  
  \subparagraph{Extracting a Revealed Preference Bound} By Lemma 3 of \citep{krishnamurthy2020identifying}, we can represent a limit $\bar{\kcov}^{-1}$ on the radar's precision of target $i$ measurement, $\kcov_{\stime}^{* -1}(\lconst,\varti)$ as the simple linear inequality $\lconst ' \varti \leq 1$, i.e., 
 \[\lconst'\varti \leq 1 \Longleftrightarrow  \kcov_{\stime}^{* -1}(\lconst,\varti) \leq \bar{\kcov}^{-1}\] 
 where the constant $1$ bound is taken without loss of generality. The key idea behind this equivalence is to show the asymptotic precision $\kcov^{* -1}_n(\cdot,\varti)$ is monotone increasing in the first argument $\lconst$ using the information Kalman filter formulation. Then, we can represent a constraint on the radar's average precision over measurements of all targets as
 \begin{equation}
 \label{eq:linc_rec}
 \lconst ' (\sum_{i=1}^{\numagents} \varti) \leq 1
\end{equation}
 Thus, we recover a direct correspondence between the radar's average measurement precision and the linear inequality constraint in \eqref{eq:MCR}. Thus, again, "collective rationality" \eqref{eq:MCR} on the part of the UAV network can directly be interpreted as the high-level constrained multi-objective optimization \eqref{moo_int}. 

The recovery of this linear constraint \eqref{eq:linc_rec} from the de-coupled Kalman filter gives a clear correspondence between the filtering dynamics and the high-level objective constraint \eqref{moo_int}. However, this de-coupled Kalman filtering scheme is idealized and simplified; next we outline a more sophisticated multi-target tracking algorithm which is widely used in practice \cite{fortmann1980multi}, \cite{rezatofighi2015joint}, and show the same recovery of the linear constraint \eqref{eq:linc_rec}. 

\paragraph{Joint Probabilistic Data Association Filter}

The joint probabilistic data association filter (JPDAF) operates under the regime where $n$ measurements $y_k^j, j \in [n]$ \eqref{lin_gaus} of $m$ targets are obtained, and it is not known which measurements correspond to which target. See \cite{bar1995multitarget} for clarification of any details.

\subparagraph{Filter Dynamics} Define the empirical validation matrix $\Omega = [\omega_{jt}, j\in[n],t\in\{0,\dots,m\}$, with $\omega_{jt} = 1$ if measurement $j$ is in the \textit{validation gate} of target $t$, and $0$ otherwise. It is common to let the $t=0$ index correspond to "none of the targets". 

Now we construct an object $\theta$ known as the "joint association event", as 
\[\theta = \bigcap_{j=1}^m \theta_{jt_j}\]
where 
\begin{itemize}
    \item[-] $\theta_{jt}$ represents the event that measurement $j$ originated from target $t$
    \item[-] $t_j$ is the index of the target which measurement $j$ is associated with in the event under consideration
\end{itemize}
So, $\theta$ can represent any possible set of associations between measurements and targets.

Then, we can form the \textit{event matrix} 
\[\hat{\Omega}(\theta) = [\hat{\omega}_{jt}]\]
where 
\[\hat{\omega}_{jt} = \begin{cases}
    1, \, \theta_{jt} \in \theta \\
    0, \, \textrm{else}
\end{cases}\]
$\hat{\Omega}(\theta)$ is thus the indicator matrix of measurement-target associations in event $\theta$. 

We say an event $\theta$ is a \textit{feasible association event} if
\begin{enumerate}
    \item a measurement is associated to only one source,
    \begin{equation}
        \label{eq:msource}
        \sum_{t=0}^m \hat{\omega}_{jt}(\theta) = 1, \quad \forall \,j\in[n]
    \end{equation}
    \item at most one measurement originates from each target, 
    \begin{equation}
        \label{eq:tdi}
        \delta_t(\theta) := \sum_{j=1}^n \hat{\omega}_{jt}(\theta) \leq 1, \quad \forall \, t\in [m]
    \end{equation}
\end{enumerate}
Denote by $\Theta$ the set of all feasible events.

The binary variable $\delta_t(\theta)$ is known as the "target detection indicator" since it indicates whether, in event $\theta$, a measurement $j$ has been associated to target $t$. We may also define a "measurement association indicator" 
\begin{equation}
\label{eq:mai}
    \tau_j(\theta) := \sum_{t=1}^m \hat{\omega}_{jt}(\theta)
\end{equation}
which indicates if a particular measurement $j$ is associated with a target $t$. Note the difference between \eqref{eq:mai} and \eqref{eq:msource}; the latter sums from $0$ to include the possibility of a measurement being assigned to "no target", i.e., clutter, while the former sums from $1$, indicating if the measurement has been assigned to an actual target. 

Using these definitions we can write the number of false (unassociated) measurements in event $\theta$ as 
\begin{equation}
    \label{eq:fm}
    \phi(\theta) := \sum_{j=1}^n[1-\tau_j(\theta)]
\end{equation}

Using these preliminary concepts, the JPDAF can be formulated by first deriving the posterior probability of joint-association events given the measured data, then incorporating this into a standard filtering scheme akin to the Kalman filter. The filtering can be done in an uncoupled or coupled manner; the former assumes target measurements are independently distributed, and the latter is capable of correlations in target state estimation errors. 

\textit{Uncoupled Filtering}: Now given a particular feasible joint-association event $\theta_k \in \Theta$, and letting $\delta_t, \tau_j, \phi$ be shorthand for \eqref{eq:tdi}, \eqref{eq:mai}, \eqref{eq:fm}, respectively evaluated at $\theta_k$, \citep{bar1995multitarget} derives the posterior probability $P(\theta_k| \{y_k^j\}_{j=1}^n)$, under the uncoupled assumption, as 
\begin{align}
\begin{split}
\label{eq:jointpost}
&P(\theta_k | \{y_k^j\}_{j=1}^n) \propto \frac{\phi!}{m_k!}\mu_{F}(\phi)V^{-\phi}\prod_{j}[f_{tj}(y_k^j)]^{\tau_j} \prod_t (P^t_D)^{\delta_t}(1-P^t_D)^{1-\delta_t}
\end{split}
\end{align}
where $P_D^t$ is the detection probability of target $t$, $m_k = n-\phi$, and 
\[f_{tj}(y_k^j) = \gaussN(y_k^j; \hat{y}^{t_j}_{k|k-1}, S^{t_j}_k)\]
with $ \hat{y}^{t_j}_{k|k-1}$ the predicted measurement for target $t_j$ in the previous iteration of the filter, and $S^{t_j}_k$ the associated innovation covariance matrix. $\mu_{F}(\phi)$ is the probability mass function governing the number of false measurements $\phi$, and such measurements not associated with a target are assumed uniformly distributed in the surveillance region of volume $V$.

Given, this posterior probability the uncoupled filter proceeds by separately filtering each target state independently. For brevity we do not introduce this filtering process, but do so for the more sophisticated and robust \textit{coupled} filter.

\textit{Coupled Filtering}:
Given a particular feasible joint-association event $\theta_k \in \Theta$, and letting $\delta_t, \tau_j, \phi$ be shorthand for \eqref{eq:tdi}, \eqref{eq:mai}, \eqref{eq:fm}, respectively evaluated at $\theta_k$, \citep{bar1995multitarget} derives the posterior probability $P(\theta_k| \{y_k^j\}_{j=1}^n)$  as 
\begin{align}
\begin{split}
\label{eq:jointpost}
P(\theta_k | \{y_k^j\}_{j=1}^n) \propto &\frac{\phi!}{m_k!}\mu_{F}(\phi)V^{-\phi}f_{t_{j_1}, t_{j_2}, \dots}(y_k^j, j:\tau_j=1) \prod_t (P^t_D)^{\delta_t}(1-P^t_D)^{1-\delta_t}
\end{split}
\end{align}
where here $f_{t_{j_1}, t_{j_2}, \dots}$ is the joint pdf of the measurements of the targets under consideration, and $t_{j_i}$ is the target which $y_{k}^{j_i}$ is associated in event $\theta_k$. 
Now we introduce the Joint Probabilistic Data Association Coupled Filter (JPDACF) state estimation and covariance update. 

We form the stacked state vector of predicted states, and associated covariance, as 
\[\hat{x}_{k|k-1} = \begin{bmatrix} \hat{x}^1_{k|k-1} \\ 
\vdots \\
\hat{x}^m_{k|k-1}
\end{bmatrix}\]

\[P_{k|k-1} = \begin{bmatrix}   
P^{1\,1}_{k|k-1}\, \dots \, P^{1\,m}_{k|k-1} \\ \vdots \quad \quad\quad \quad \vdots \\ 
P^{m\,1}_{k|k-1}\, \dots \, P^{m\,m}_{k|k-1}
\end{bmatrix}\]

where $P^{t_1\,t_2}$ is the cross-covariance between targets $t_1$ and $t_2$. The coupled filtering is done as follows:
\[\hat{x}_{k|k} = \hat{x}_{k|k-1} + W_k \sum_{\theta} P(\theta | \{y_k^j\}_{j=1}^n)[\boldsymbol{y}_k(\theta) - \hat{y}_{k|k-1}]\]
where 
\[\boldsymbol{y}_k(\theta) = \begin{bmatrix} 
y_k^{j_1(\theta)}\\ \vdots \\ y_k^{j_m(\theta)}
\end{bmatrix}\]
and $j_i(\theta)$ is the measurement associated with target $i$ in event $\theta$. The filter gain $W_k$ is given by 
\[W_k = P_{k|k-1} \hat{C}_k'\left[\hat{C}_k P_{k|k-1} \hat{C}_k' + \hat{R}_k\right]^{-1}\] 
where
\begin{align*}
    &\hat{C}_k = \textrm{diag}\left[\delta_{1}(\theta) C_k^1, \dots, \delta_m(\theta) C_k^m\right] \\
    &\hat{R}_k = \textrm{diag}\left[R_k^1, \dots,R_k^m\right]
\end{align*}
are the block diagonal measurement and noise covariance matrices. The binary detection indicator variables $\delta_i(\theta)$ accounts for the possibility of a measurement not being associated to target $i$. The predicted stacked measurement vector is 
\[\hat{y}_{k|k-1} = \hat{C}_k\hat{x}_{k|k-1} = \hat{C}_k\hat{A}_{k-1} \hat{x}_{k-1}\]
with $\hat{A}_{k-1} = \textrm{diag}[A_{k-1}^1,\dots,A_{k-1}^m]$ the block diagonal state update matrix. 

The covariance of the updated state is given as
\begin{equation}
\label{eq:CovUpd}
P_{k|k} = P_{k|k-1} + [1-\psi_{0}]W_k \hat{S}_k W_k' + \tilde{P}_k
\end{equation}
where $\hat{S}_k = \hat{C}_k P_{k|k-1} \hat{C}_k' + \hat{R}_k$ is the innovation covariance,
\[\psi_{jt} := \sum_{\theta : \theta_{jt} \in \theta}P(\theta | \{y_k^j\}_{j=1}^n)\]
and $\psi_0 := \sum_{j=1}^m \psi_{j 0}$ is the probability that no measurements arise from targets. $\tilde{P}_k$ is the spread of the innovation terms:
\[\tilde{P}_k := W_k \tilde{S}_k W_k'\]
with
\[ \tilde{S}_k = \begin{bmatrix} \sum_{j=1}^{m_k} \psi_{j1}\left[y_k^1 - \hat{x}^1_{k|k-1}\right]\cdot \left[y_k^1 - \hat{x}^1_{k|k-1}\right]'  -\nu_{1,k} \nu_{1,k} '\\ \vdots \\ \sum_{j=1}^{m_k} \psi_{jm}\left[ y_k^m - \hat{x}^m_{k|k-1}\right]\cdot\left[ y_k^m - \hat{x}^m_{k|k-1}\right]' - \nu_{m,k}\nu_{m,k}'
\end{bmatrix} \] 
and
\[\nu_{i,k} = \sum_{j=1}^{m_k}\psi_{ji}\left[ y_k^i - \hat{x}^i_{k|k-1}\right]\]

\subparagraph{Extracting a Revealed Preference Bound} The crucial observation is that, as in the Kalman filter algebraic Riccati equation \eqref{eq:ARE}, the covariance \eqref{eq:CovUpd} is monotone decreasing in $\beta_t^i$ for all $i$, since this corresponds to increasing $\hat{R}_k$ for fixed $k$. Thus, the asymptotic measurement precision (inverse of asymptotic predicted covariance) is monotone \textit{non-decreasing} in $\beta_t^i$, and by the same reasoning as Lemma 3 of \citep{krishnamurthy2020identifying}, we may derive the equivalence 
\[\alpha_t\left(\sum_{i=1}^M \beta_t^i \right) \leq 1 \Longleftrightarrow \lim_{k\to\infty}P_{k|k}^{-1}(\alpha_t, \{\beta_t^i\})\leq \hat{P}^{-1} \]
Thus, we again have that the constraint $\alpha_t\left(\sum_{i=1}^M \beta_t^i \right) \leq 1$ is a natural representation for a bound on the average measurement precision.

\section{Detection of Coordination}
\label{sec:cdet}
In Section~\ref{sec:mtf} we showed how a notion of coordination, corresponding to linearly constrained multi-objective optimization, arises naturally from several standard multi-target filtering algorithms. In this section, we illustrate how to detect coordination in UAV networks using the microeconomic revealed preference tools in Section~\ref{sec:imoo}. We first consider deterministic detection, which is a straightforward application of the results in Section~\ref{sec:imoo}, then extend this to optimal statistical detection when UAV maneuvers are observed in noise. 
\subsection{Deterministic Coordination Detection}
We take $\varti > 0 \ \forall t \in [T], i \in [\numagents]$, i.e., each UAV always has a non-zero process noise. 
Then by Lemma 1 in \citep{snow2023statistical}, \eqref{def:coord_eq} is equivalent to 
\begin{align}
\begin{split}
\label{eq:moo}
    \{\varti\}_{i=1}^{\numagents} \in &\arg\max_{\{\argo^i\}_{i=1}^{\numagents}} \sum_{i=1}^{\numagents} \sw^i f^i(\argo^i) \ \ s.t. \ \lconst' (\sum_{i=1}^{\numagents}\argo^i ) \leq 1
\end{split}
\end{align}
for any $\sw \in \psimplex$.

Recall that we are interested in the \textit{inverse} multi-objective optimization problem. The equivalence between \eqref{eq:moo} and \eqref{def:coord_eq} allows us to directly utilize the microeconomic result Theorem~\ref{thm:cherchye1}, such that detecting coordination is equivalent to solving the linear program \eqref{af_ineq}. Furthermore, we can reconstruct feasible utility functions which rationalize the dataset as \eqref{eq:Utrec}. 
This procedure for detection of coordination and utility function reconstruction is illustrated in Algorithm~\ref{alg:Cdet}.

\begin{algorithm}
\caption{Detecting Coordination}
\begin{algorithmic}[1]
\State Record the time-indexed dataset of radar waveforms and UAV network responses $\boldsymbol{\beta} = \{\alpha_t,\{ \beta_t^i\}_{i=1}^M, t\in[T]\}$. 
\If {$\exists u_j^i, \lambda_j^i : \, \forall s,t\in[T],i\in[M]: u_s^i - u_t^i - \lambda_t^i\alpha_t'[\beta_s^i-\beta_t^i]\leq 0$}
    \State Declare coordination present
    \State Reconstruct feasible utility functions $U^i$ as 
    $U^i(\cdot) = \min_{t\in[T]}[u_t^i + \lambda_t^i\alpha_t'[\cdot - \beta_t^i]$
    \State $\Rightarrow \exists \mu \in \simplex : \, \{\beta_t^i\}_{i=1}^M \in \arg\max_{\{\gamma^i\}_{i=1}^M}\sum_{i=1}^M\mu_i U^i(\gamma^i) \, : \alpha_t'\left(\sum_{i=1}^M\gamma^i\right) \leq 1$
\EndIf
\end{algorithmic}
\label{alg:Cdet}
\end{algorithm}

\begin{figure}[t]
\label{fig:utilities}
        \centering
        \begin{subfigure}[]{0.3\textwidth} 
            \centering
            \includegraphics[width=\textwidth]{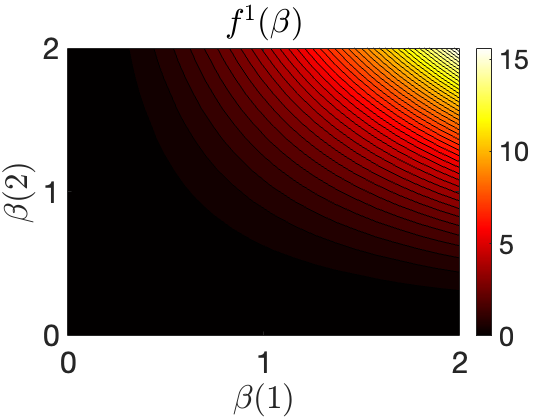}
            \caption[]%
            {{\scriptsize $f^1(\beta) = \textrm{det}(Q(\beta))$}}    
            \label{fig:T_util1}
        \end{subfigure}
        \hfill
        \begin{subfigure}[]{0.3\textwidth}   
            \centering 
            \includegraphics[width=\textwidth]{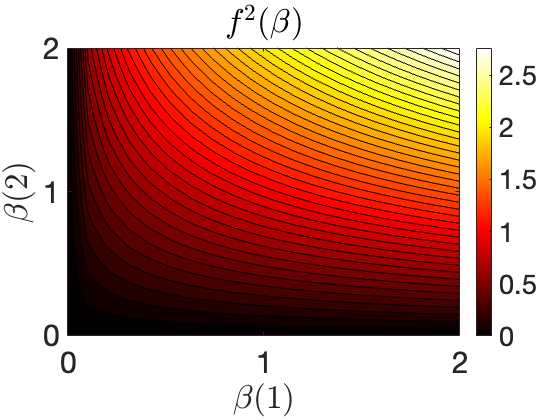}
            \caption[]%
            {{\scriptsize $f^2(\beta) = \textrm{Tr}(Q(\beta))$}}    
            \label{fig:T_util2}
        \end{subfigure}
        \hfill
        \begin{subfigure}[]{0.3\textwidth}   
            \centering 
            \includegraphics[width=\textwidth]{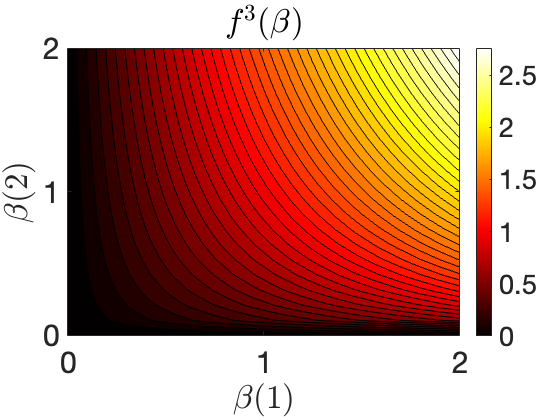}
            \caption[]%
            {{\scriptsize $f^3(\beta) = \sqrt{\beta(1)}\beta(2)$}}    
            \label{fig:T_util3}
        \end{subfigure}
        \hfill
        \vskip\baselineskip
        
        \begin{subfigure}[]{0.3\textwidth}  
            \centering 
            \includegraphics[width=\textwidth]{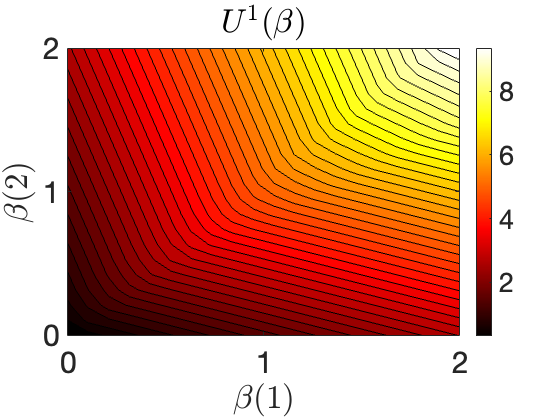}
            \caption[]%
            {{\scriptsize $\min_{t\in[10]}[u_t^1 + \lambda_t^1\alpha_t'[\cdot - \beta_t^1]$}}    
            \label{fig:util1}
        \end{subfigure}
        \hfill
        \begin{subfigure}[]{0.3\textwidth}   
            \centering 
            \includegraphics[width=\textwidth]{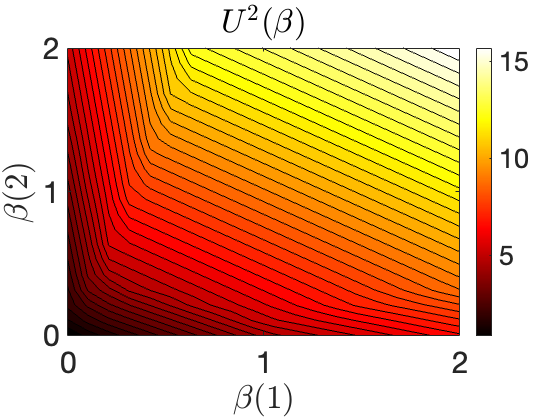}
            \caption[]%
            {{\scriptsize $\min_{t\in[10]}[u_t^2 + \lambda_t^2\alpha_t'[\cdot - \beta_t^2]$}}    
            \label{fig:util2}
        \end{subfigure}
        \hfill
        \begin{subfigure}[]{0.3\textwidth}   
            \centering 
            \includegraphics[width=\textwidth]{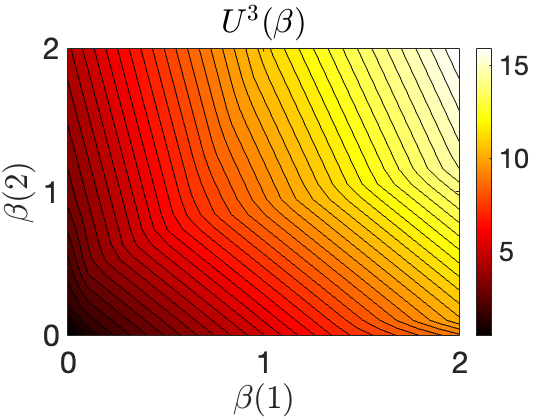}
            \caption[]%
            {{\scriptsize $\min_{t\in[10]}[u_t^3 + \lambda_t^3\alpha_t'[\cdot - \beta_t^3]$}}    
            \label{fig:util3}
        \end{subfigure}
        \caption[]
        {\small $f^i(\beta)$ is the true objective function of the $i$'th radar, inducing the responses $\{\beta_t^i\}_{t=1}^{10}$. $U^i(\beta)$ is the reconstructed objective function for radar $i$, computed using the dataset $\boldsymbol{\beta} = \{\lconst, \{\varti\}_{i=1}^M, t\in[T]\}$ and 
        \eqref{eq:Utrec}.}
        \label{fig:utilities}
\end{figure}

\subparagraph{Numerical Example}
Here we provide a numerical example for the deterministic coordination detection and utility reconstruction procedure outlined in Algorithm~\ref{alg:Cdet}. We consider $\numagents = 3$ targets, and acquire data over $T=10$ time-steps. The radar waveform measurement covariance eigenvalue vector $\alpha_t$, and each target maneuver vector $\beta_t^i$, are taken to be two-dimensional. The three targets are taken to have the following simple utility functions:
\begin{align}
\begin{split}
\label{eq:NumUtils}
    f^1(\beta) &= \textrm{det}(Q(\beta))^2 = \beta(1)^2\beta(2)^2 \\
    f^2(\beta) &= \sqrt{\beta(1)}\beta(2) \\
    f^3(\beta) &= \beta(1)\sqrt{\beta(2)}
\end{split}
\end{align}

We then generate the vectors $\alpha_t,\, \beta_t^i$, with $\mu^1=0.4, \mu^2 = 0.4, \mu^3 =0.3$, as follows:
\begin{itemize}
    \item $\alpha_t \sim U[0.1,1.1]^2$
    \item $\{\beta_t^i\}_{i=1}^M \in \arg\max_{\{\gamma^i\}_{i=1}^M}\sum_{i=1}^3 \mu^i f^i(\gamma^i) \quad \textrm{s.t. } \alpha_t'(\sum_{i=1}^3\gamma^i) \leq 1$
\end{itemize}
Thus, the target responses $\{\beta_t^i\}$ satisfy our notion of coordination (multi-objective optimization). Then, implementing Algorithm~\ref{alg:Cdet}, we confirm that the linear program \eqref{af_ineq} has a feasible solution, indicating the presence of multi-objective optimization, and we may reconstruct feasible utility functions. Reconstructed utility functions are illustrated in Figure~\ref{fig:utilities}. Notice that the reconstructed utility functions match the relative  profiles of the true utility functions, and do so while being concave.

\subsection{Statistical Detection of Coordination}
Recall that thus far we have considered only deterministic UAV $i$ dynamics $\varti$. We now consider the case when these measured responses are corrupted by noise. 

We introduce a statistical detector for determining whether these \textit{noisy} responses are consistent with multi-objective optimization, with theoretical guarantees on Type-I error.


 Let $\barbeta$ denote the dataset when the radar responses are  observed in noise:
\begin{equation}
\label{eq:dataset}
    \barbeta = \{\lconst, \nrespi , t \in [T], i \in [\numagents]\}
\end{equation}
where $\nrespi = \respi + \noisei_t$, and $\noise^i_t$ are independent random variables generated according to distributions $\ndistr^i_t$. 
We propose a statistical detector to optimally determine if the responses are consistent with Pareto optimality \eqref{eq:MOP}. Define \\
$H_0$: null hypothesis that the dataset \eqref{eq:dataset} arises from the optimization problem \eqref{def:coord_eq} for all $t\in[T]$. \\
$H_1$: alternative hypothesis that the dataset \eqref{eq:dataset} does not arise from the optimization problem \eqref{def:coord_eq} for all $t\in[T]$. 

There are two possible sources of error: \\
\textbf{Type-I error}: Reject $H_0$ when $H_0$ is valid.\\
\textbf{Type-II error}: Accept $H_0$ when $H_0$ is invalid. 

We formulate the following test statistic $\optstat$, as a function of $\barbeta$, to be used in the detector: 
\begin{equation}
\label{eq:optstat}
\optstat = \max_i \optstati
\end{equation}
where $\optstati$ is the solution to:
\begin{align}
\begin{split}
\label{eq:LP}
&\min \stati : \exists \ u_t^i > 0, \lambda_t^i > 0 : \, u_s^i - u_t^i - \lambda_t^i \lconst'(\bar{\var}^i_s - \bar{\var}_t^i) - \lambda_t^i \stati \leq 0 
\end{split}
\end{align}
Form the random variable $\rvtest$ as 
\begin{align}
\begin{split}
\label{eq:psimax}
&\rvtest =\max_{i, \,t\neq s}[\lconst'(\noisei_t - \noisei_s)]
\end{split}
\end{align}
Then we propose the following statistical detector (with $\gamma \in (0,1)$):
\begin{equation}
\label{eq:stat_test}
\int_{\optstat}^{\infty} f_{\rvtest}(\psi)d\psi 
\begin{cases}
\geq \gamma \Rightarrow H_0 \\ < \gamma \Rightarrow H_1
\end{cases}
\end{equation}
where $f_{\rvtest}(\cdot)$ is the probability density function of $\rvtest$. Let $\fcdf$ be the cdf of $\rvtest$ and $\fccdf$ be the complementary cdf of $\rvtest$. Then we have the following guarantees:

\begin{theorem}
\label{thm:stat_det}
Consider the noisy dataset $\eqref{eq:dataset}$, and suppose \eqref{eq:LP} has a feasible solution. Then 
\begin{enumerate}
    \item The following null hypothesis implication holds:
    \begin{equation}
    \label{eq:H0equiv} H_0 \subseteq \bigcap_{i \in [M]} \{\optstati \leq \rvtesti\} \end{equation}
    \item The probability of Type-I error (false alarm) is 
    \[ \Popt(H_1 | H_0) = \Prob(\fccdf(\optstat) \leq \gamma \ | H_0) \leq \gamma \]

    \item The optimizer $\optstat$ yields the smallest Type-I error bound:
    \begin{align*}
    \begin{split}
    &\Prob_{\bar{\boldsymbol{\Phi}}(\barbeta)}(H_1 | H_0) \geq \Popt(H_1 | H_0) \quad \forall \bar{\boldsymbol{\Phi}}(\barbeta) \in [\optstat, \rvtest]
    \end{split}
    \end{align*}
    
\end{enumerate}
\end{theorem}

\begin{proof}
    See \citep{snow2023statistical}
\end{proof}

The motivation for this detector is that it allows one to quantify a strict upper bound on the probability of Type-I error; the specific choice of threshold $\gamma$ is left to the designer, and may vary depending on application criteria. 

\begin{algorithm}[h!]
\caption{Detecting Multi-Objective Optimization}
\begin{algorithmic}[1]
    \For{l=1:L}
        \For{i=1:M}
             \State simulate $\boldsymbol{\noise}^i_l = [\noise^i_1,\dots,\noise^i_N]^{(l)}, \quad \noise^i_t \sim \Lambda_t^i$
        \EndFor
        \State Compute $\Psi^{l} := \max_i \{\max_{t\neq s}[\lconst(\noise_t^i - \noise_s^i)]\}$
    \EndFor
    \State Compute $\hat{F}_{\Psi}(\cdot)$ from $\{\Psi^l\}_{l=1}^L$
\State Record radar network response $\barbeta$ to the probe $\lconst$
\State Solve \eqref{eq:optstat} for $\optstat$
\State Save $\mathcal{P} := \{\hat{u}_t^i, \hat{\lambda}_t^i, t \in [T], i \in [\numagents]\}$ such that
\begin{align*}
    \hat{u}_s^i - \hat{u}_t^i - \hat{\lambda}_t^i \lconst'(\bar{\var}^i_s - \bar{\var}_t^i) - \hat{\lambda}_t^i \optstati \leq 0 \ \forall i \in [\numagents]
\end{align*}
\State Implement detector \eqref{eq:stat_test} as
\begin{equation}
    1  - \hat{F}_{\Psi}(\optstat) \begin{cases}
         > \gamma \Rightarrow H_0 \\
         \leq \gamma \Rightarrow H_1
    \end{cases}
\end{equation}
\end{algorithmic}
\label{alg:MOOdet}
\end{algorithm}

In practice one would likely not have access to the true density function  $f_{\Psi}(\cdot)$. However, it is typical to assume some structure on the additive noise process $\{\Lambda_t^i, t\in[T]\}_{i\in[\numagents]}$ such as Gaussianity. Thus, under such an assumption one can compute an approximation $\hat{F}_{\Psi}(\cdot)$ of the cumulative distribution function $F_{\Psi}(\cdot)$, then implement the statistical detector using this. Algorithm~\ref{alg:MOOdet} provides such an implementation of the statistical detector \eqref{eq:stat_test}.

\begin{figure}
\centering
  \includegraphics[width=0.6\linewidth,scale=1]{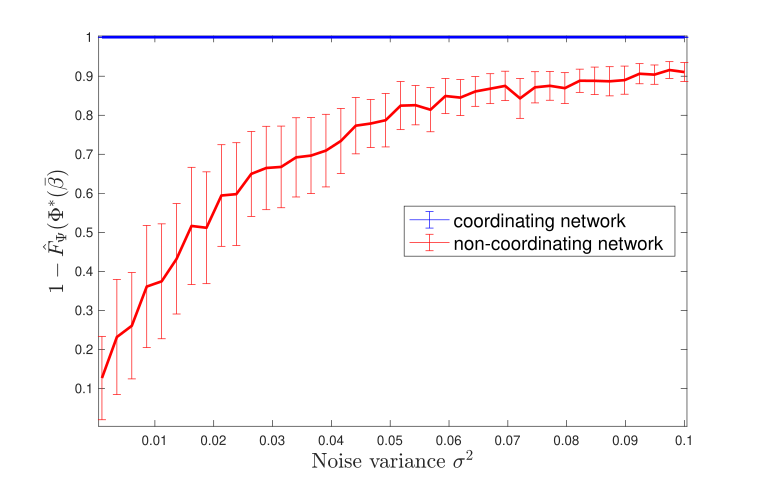}
  \caption{ Statistic $1  - \hat{F}_{\Psi}(\optstat)$ as a function of variance of the noise distribution $\Lambda_t$. Error bars represent one standard deviation within the dataset produced by 300 Monte-Carlo simulations. Higher $1  - \hat{F}_{\Psi}(\optstat)$ corresponds to higher likelihood of radar network coordination in the statistical detector \eqref{eq:stat_test}.}
  \label{fig:StatDet}
\end{figure}

\subparagraph{Numerical Example}
Here we investigate the empirical behavior of the statistic $1  - \hat{F}_{\Psi}(\optstat)$ under both $H_0$ and $H_1$. We generate the statistic from the procedure outlined in Algorithm 1, with $L = 500$, $M = 3$, $T=10$. The probe signal $\lconst \in \reals^2$ is generated randomly as $\lconst \sim U[0.1,1.1]^2 $, i.e. each element of $\lconst$ is generated as an independent uniform random variable on the interval [0.1,1.1]. To simulate a UAV network, the responses $\{\varti\}_{i=1}^{\numagents}$ are taken as solutions to the multi-objective optimization \eqref{def:coord} with objective functions given by \eqref{eq:NumUtils}, and $\sw^1=\sw^2=\sw^3=1/3$. Then noisy responses $\{\bar{\var}_t^i\}_{i=1}^{\numagents}$ are obtained by adding i.i.d. Gaussian noise $\noisei_t \sim \Lambda_t = \gaussN(0,\sigma^2)$. The blue line in Figure~\ref{fig:StatDet} displays the resultant empirical statistic $1  - \hat{F}_{\Psi}(\optstat)$ as a function of noise variance. 
To simulate a non-coordinating radar network, we generate each response $\varti \sim U[0,1]^2$ independently, and similarly add Gaussian measurement noise $\noisei_t \sim \Lambda_t = \gaussN(0,\sigma^2)$. The red line in Figure~\ref{fig:StatDet} is the empirical statistic $1  - \hat{F}_{\Psi}(\optstat)$ under these circumstances, when no coordination is present.

Let us interpret the simulation results displayed in Figure~\ref{fig:StatDet}. 
Observe that the statistic $1  - \hat{F}_{\Psi}(\optstat)$ is a constant value of 1 for the noise variance range simulated. This validates our choice that the null hypothesis $H_0$ (coordination) should be chosen once the statistic surpasses a threshold. Furthermore, it indicates the strength of the statistical detector's ability to filter noise and correctly determine that coordination is present. However, as the noise variance increases the probability of type-II error (determining $H_0$ under $H_1$) grows, since the statistic $1  - \hat{F}_{\Psi}(\optstat)$ becomes more likely to surpass a given threshold $\gamma \in (0,1)$. This is an unavoidable consequence, within any statistical detection scheme, of the degraded ability to differentiate coordination vs non-coordination as the noise power grows. However, the particular behavior displayed in Figure~\ref{fig:StatDet} gives insight into the control of type-II error, since one may choose the threshold $\gamma$ to be arbitrarily close to one, in this small noise regime, such that the probability of type-I error remains constant but that of type-II error is diminished. 

This section presented techniques for both deterministic and statistical detection of coordination in UAV networks. These techniques exploit the microeconomic revealed preference results in Section~\ref{sec:imoo} and the abstract correspondence between UAV dynamics and linearly constrained multi-objective optimization in Sections~\ref{sec:cso}, \ref{sec:mtf}.

\section{Conclusion}

We have investigated the mathematical properties of multi-objective optimization and inverse multi-objective optimization, and presented a microeconomic technique for performing the latter. We have demonstrated how this can be applied in a UAV network coordination detection scheme, by utilizing radar tracking signals. This methodology is more abstract than traditional electronic warfare procedures, and thus allows for a concise encapsulation of the above stated problem and algorithmic solution. We also show how this abstract formulation can be recovered by several specific multi-target filtering algorithms and by specifications of radar waveform design. However, the application of the presented methodology is not limited to these cases, and can find use in a variety of inverse multi-objective optimization settings.

\bibliographystyle{plain}
\bibliography{Bibliography.bib}

\begin{thebibliography}{10}

\bibitem{afriat1967construction}
Sydney~N Afriat.
\newblock The construction of utility functions from expenditure data.
\newblock {\em International economic review}, 8(1):67--77, 1967.

\bibitem{bar2004estimation}
Yaakov Bar-Shalom, X~Rong Li, and Thiagalingam Kirubarajan.
\newblock {\em Estimation with applications to tracking and navigation: theory algorithms and software}.
\newblock John Wiley \& Sons, 2004.

\bibitem{bar1995multitarget}
Yaakov Bar-Shalom and Xiao-Rong Li.
\newblock {\em Multitarget-multisensor tracking: principles and techniques}, volume~19.
\newblock YBs Storrs, CT, 1995.

\bibitem{chen2020toward}
Wu~Chen, Jiajia Liu, Hongzhi Guo, and Nei Kato.
\newblock Toward robust and intelligent drone swarm: Challenges and future directions.
\newblock {\em IEEE Network}, 34(4):278--283, 2020.

\bibitem{cherchye2011revealed}
Laurens Cherchye, Bram De~Rock, and Frederic Vermeulen.
\newblock The revealed preference approach to collective consumption behaviour: Testing and sharing rule recovery.
\newblock {\em The Review of Economic Studies}, 78(1):176--198, 2011.

\bibitem{chiappori2009microeconomics}
P-A Chiappori and Ivar Ekeland.
\newblock The microeconomics of efficient group behavior: Identification 1.
\newblock {\em Econometrica}, 77(3):763--799, 2009.

\bibitem{fortmann1980multi}
Thomas~E Fortmann, Yaakov Bar-Shalom, and Molly Scheffe.
\newblock Multi-target tracking using joint probabilistic data association.
\newblock In {\em 1980 19th IEEE Conference on Decision and Control including the Symposium on Adaptive Processes}, pages 807--812. IEEE, 1980.

\bibitem{kershaw1994optimal}
David~J Kershaw and Robin~J Evans.
\newblock Optimal waveform selection for tracking systems.
\newblock {\em IEEE Transactions on Information Theory}, 40(5):1536--1550, 1994.

\bibitem{krishnamurthy2020identifying}
Vikram Krishnamurthy, Daniel Angley, Robin Evans, and Bill Moran.
\newblock Identifying cognitive radars-inverse reinforcement learning using revealed preferences.
\newblock {\em IEEE Transactions on Signal Processing}, 68:4529--4542, 2020.

\bibitem{li2003survey}
X~Rong Li and Vesselin~P Jilkov.
\newblock Survey of maneuvering target tracking. part i. dynamic models.
\newblock {\em IEEE Transactions on aerospace and electronic systems}, 39(4):1333--1364, 2003.

\bibitem{marden2014achieving}
Jason~R Marden, H~Peyton Young, and Lucy~Y Pao.
\newblock Achieving pareto optimality through distributed learning.
\newblock {\em SIAM Journal on Control and Optimization}, 52(5):2753--2770, 2014.

\bibitem{miettinen2012nonlinear}
Kaisa Miettinen.
\newblock {\em Nonlinear multiobjective optimization}, volume~12.
\newblock Springer Science \& Business Media, 2012.

\bibitem{quintero2010optimal}
Steven~AP Quintero, Francesco Papi, Daniel~J Klein, Luigi Chisci, and Joao~P Hespanha.
\newblock Optimal uav coordination for target tracking using dynamic programming.
\newblock In {\em 49th IEEE Conference on Decision and Control (CDC)}, pages 4541--4546. IEEE, 2010.

\bibitem{ruadulescu2020multi}
Roxana R{\u{a}}dulescu, Patrick Mannion, Diederik~M Roijers, and Ann Now{\'e}.
\newblock Multi-objective multi-agent decision making: a utility-based analysis and survey.
\newblock {\em Autonomous Agents and Multi-Agent Systems}, 34(1):10, 2020.

\bibitem{rezatofighi2015joint}
Seyed~Hamid Rezatofighi, Anton Milan, Zhen Zhang, Qinfeng Shi, Anthony Dick, and Ian Reid.
\newblock Joint probabilistic data association revisited.
\newblock In {\em Proceedings of the IEEE international conference on computer vision}, pages 3047--3055, 2015.

\bibitem{shi2017power}
Chenguang Shi, Sana Salous, Fei Wang, and Jianjiang Zhou.
\newblock Power allocation for target detection in radar networks based on low probability of intercept: A cooperative game theoretical strategy.
\newblock {\em Radio Science}, 52(8):1030--1045, 2017.

\bibitem{snow2023statistical}
Luke Snow and Vikram Krishnamurthy.
\newblock Statistical detection of coordination in a cognitive radar network through inverse multi-objective optimization.
\newblock {\em IEEE International Conference on Information Fusion}, 2023.

\bibitem{snow2022identifying}
Luke Snow, Vikram Krishnamurthy, and Brian~M Sadler.
\newblock Identifying coordination in a cognitive radar network--a multi-objective inverse reinforcement learning approach.
\newblock {\em International Conference on Acoustics, Speech, and Signal Processing}, 2022.

\bibitem{van2004detection}
Harry~L Van~Trees.
\newblock {\em Detection, estimation, and modulation theory, part I: detection, estimation, and linear modulation theory}.
\newblock John Wiley \& Sons, 2004.

\bibitem{wise2006uav}
Richard Wise and Rolf Rysdyk.
\newblock Uav coordination for autonomous target tracking.
\newblock In {\em AIAA Guidance, Navigation, and Control Conference and Exhibit}, page 6453, 2006.

\end{thebibliography}

\end{document}